\documentclass[sn-mathphys-num]{sn-jnl}
\usepackage{graphicx} 

\usepackage[most]{tcolorbox}
\usepackage{amsmath}
\usepackage{amssymb}
\usepackage{amsthm}
\usepackage{hyperref}
\newcommand{\R}{\mathbb{R}}

\newcommand{\norm}[1]{\left\| #1 \right \|}
\newcommand{\abs}[1]{\left | #1 \right |}
\newtheorem{theorem}{Theorem}
\newtheorem{lemma}[theorem]{Lemma}
\newtheorem{definition}[theorem]{Definition}
\newtheorem{remark}[theorem]{Remark}
\newtheorem{corollary}[theorem]{Corollary}

\DeclareMathOperator{\alt}{alt}
\DeclareMathOperator{\vol}{vol}
\DeclareMathOperator{\height}{height}

\newcommand{\conv}{\mathrm{conv}}
\newcommand{\scproduct}[2]{\left\langle#1, #2\right\rangle}

\newcommand{\spann}{\mathrm{span}}

\usepackage{xcolor}
\definecolor{darkgreen}{RGB}{0,100,0}
\definecolor{darkred}{RGB}{139,0,0}

\title{Generalized altitudes and their bounds
}

\author*[1]{\fnm{Hana} \sur{Dal Poz Kou\v{r}imsk{\'a}}}\email{hana.dal.poz.kourimska@uni-potsdam.de}

\author[2]{\fnm{Mathijs} \sur{Wintraecken}}\email{mathijs.wintraecken@inria.fr}

\affil*[1]{\orgdiv{Institut f{\"u}r Mathematik der Universit{\"a}t Potsdam}, 
\orgname{
University of Potsdam, Campus Golm}, \orgaddress{\city{Golm}, \postcode{D-14476}, \country{Germany}}, \url{https://orcid.org/0000-0001-7841-0091}}

\affil[2]{\orgdiv{Centre Inria d'Université Côte d'Azur}, \orgaddress{\city{Sophia-Antipolis}, \postcode{06902}, \country{France}}, \url{https://orcid.org/0000-0002-7472-2220}}

\abstract{
We introduce generalized altitudes of a simplex, extending the usual vertex-to-opposite-face altitude to arbitrary pairs of opposite faces. These quantities encode the relative position of the affine spans of such faces and yield a uniform formula for the angle between them. We also derive an equivalent algebraic expression in terms of generalized cross products and Gram determinants, linking the construction to standard determinant-based tools. Finally, we prove that every generalized altitude is bounded below by a quantity controlled by the ordinary height of the simplex. Thus, classical height or thickness assumptions imply control over this broader family of geometric quantities. The results provide a compact framework for studying simplex quality and are motivated by applications to triangulation criteria for Riemannian manifolds.
}

\begin{document}

\maketitle

\section{Introduction}
Well-shaped simplices are essential for numerical accuracy in several areas of computational mathematics, including the numerical solution of partial differential equations by finite element methods
\cite{synge1957hypercircle, babuska1976angle, babuvvska1978error, jamet1976estimations, krizek1992maximum, shewchuk2002good}
and manifold meshing
\cite{cheng2005manifold, cheng2012delaunay}.
In finite element methods, poorly shaped simplices can lead to large discretization errors and ill-conditioned stiffness matrices. In manifold meshing, simplex quality controls, among other things, the angle between a simplex whose vertices lie on the manifold and the tangent space of the manifold at, or near, those vertices
\cite{whitney1957}.

Several quality measures are used to quantify how well shaped a simplex is. The choice of measure often depends on the field, and sometimes on the specific application. Common examples include the \emph{thickness}
\cite{munkres1966elementary, vavasis1996stable}, defined as the ratio between the height of the simplex---that is, its smallest altitude---and its largest edge length; the \emph{fatness}
\cite{whitney1957}, defined as the ratio between the volume of the simplex and the largest edge length raised to the dimension of the simplex; and the \emph{inradius--circumradius ratio}
\cite{cavendish}.
These measures are weakly equivalent: upper and lower bounds on one of them yield corresponding upper and lower bounds on the others. Bounds on \emph{dihedral angles} are also frequently used as a measure of simplex quality. Thickness, in particular, has proved especially convenient in the context of manifold meshing
\cite{DelaunayTriang, Localcrits}, because it can be related directly to the eigenvalues associated with the edge vectors of the simplex.

The main goal of this paper is to extend the classical notion of altitude, which is defined from a vertex to the opposite face, to arbitrary pairs of opposite faces of a simplex (Definition \ref{def:GenAlt}). We call the resulting quantities \emph{generalized altitudes}.

We first show that generalized altitudes provide a natural way to describe the angle between (the affine hulls of) a face of a simplex and its opposite face (Lemma~\ref{lem:altitude_expression}). More precisely, they lead to an explicit formula for this angle. The formula is simple, geometrically transparent, and applies uniformly to faces of arbitrary dimension.

We then relate generalized altitudes to standard tools from multilinear algebra and analysis, in particular the (generalized) cross product and the Gram determinant (Corollary~\ref{cor:altitude_Gram_determinant}). This connection gives a second, algebraic formula for generalized altitudes, and provides a bridge between the geometry of simplices and determinant-based analytic estimates.

Finally, we prove a lower bound for generalized altitudes in terms of the height of the simplex (Lemma~\ref{CL:altitute_bounded_by_height}). This estimate shows that lower bounds on the usual height of a simplex also control the generalized altitudes, and therefore gives a direct mechanism for controlling this broader notion of simplex quality.

We believe that both the geometric formula and the algebraic connections developed here are useful tools for studying simplex quality. This usefulness is reflected in the fact that the same notion arose naturally in two projects\footnote{As of 2026, these projects are still ongoing, but should be completed by the end of the year. } concerned with improved triangulation criteria for Riemannian manifolds, where these formulas became key ingredients in the proofs.


\begin{tcolorbox}[
  colback=gray!5,
  colframe=black,
  boxrule=0.5pt,
  arc=2pt
]
In the remainder of this article, we write $\sigma$ for an $n$-dimensional simplex in the $n$-dimensional Euclidean space,
\[\sigma = \{v_0,\dots,v_n\}\subseteq \R^n.\]
We assume that $n\geq 2$. For $0\leq k\leq n$, we write $\tau$ for the $k$-dimensional face of $\sigma$ consisting of the first $k+1$ vertices, and $\tau^*$ for its opposite face,
\[\tau=\{v_0,\dots,v_k\}\subseteq\sigma,\qquad \tau^* = \sigma\backslash\tau = \{v_{k+1},\dots,v_n\}\subseteq\sigma.\]
We assume $\sigma$ to be non-degenerate.
\end{tcolorbox}

\section{Generalized altitude}

The definition of altitude we introduce here is a very natural generalization of the altitude of a vertex:
\begin{definition}\label{def:GenAlt}
The \emph{altitude} of~$\tau$ (and~$\tau^*$),
	$\mathrm{alt}(\tau,\tau^*),$
	is the smallest distance between the affine hulls of~$\tau$, and of~$\tau^*$. In other words, if we denote the affine hulls of~$\tau$ and~$\tau^*$ by~$\mathcal{F}_{\tau}$ and~$\mathcal{F}_{\tau^*}$, respectively, then
	\[\mathrm{alt}(\tau,\tau^*) = \min_{u\in\mathcal{F}_{\tau}, u^*\in\mathcal{F}_{\tau^*}} d(u,u^*). \]
\end{definition} 
This definition encompasses the common notion of the {altitude of a vertex},  which is the distance from the vertex to the affine hull of the opposing facet in a simplex.
	
	

    In general, the distance between two smooth (non-intersecting) submanifolds of Euclidean space is attained along a line segment that is normal to both of them at the respective endpoints. This simple idea can be leveraged for simplices to find a concise and useful formula for the generalized altitude. To be more precise:  

    
    
	\begin{lemma}\label{lem:altitude_expression}
    Let
		\[\mathcal{V}_{\tau} = \spann\{ v_1-v_0,\dots, v_k-v_0 \}, \qquad\mathcal{V}_{\tau^*} = \spann\{v_{k+1} - v_n, \dots,v_{n-1}-v_n\},\] 
        be vector spaces parallel to the affine hulls of $\tau$ and $\tau^*$, respectively, and write $\mathcal{L}= (\mathcal{V}_{\tau}+\mathcal{V}_{\tau*})^\perp\subseteq\R^n$ for their (one-dimensional) orthogonal complement in $\R^n$.
		
        With this notation we have that, for any non-zero vector $w\in\mathcal{L}$, the altitude of~$\tau$ and~$\tau^*$ satisfies
        \begin{align}\label{eq:altitude_angle}
            \mathrm{alt}(\tau,\tau^*) = \abs{\scproduct{v_0-v_n}{\frac{w}{\norm{w}}}}.
        \end{align}
	\end{lemma}

	\begin{proof}
		The affine hulls $\mathcal{F}_{\tau}$ and $\mathcal{F}_{\tau^*}$ of, respectively, $\tau$ and $\tau^*$, can be written as
		
	\[\mathcal{F}_{\tau} = v_0 +\mathcal{V}_\tau, \qquad\text{and}\qquad \mathcal{F}_{\tau^*} = v_n+\mathcal{V}_{\tau*}.\]
	
	The altitude~$\mathrm{alt}(\tau,\tau^*)$ is the minimum distance between them,
		\[\mathrm{alt}(\tau,\tau^*) = \min_{u\in\mathcal{F}_{\tau}, u^*\in\mathcal{F}_{\tau^*}} \norm{u-u^*}= \min_{u\in\mathcal{V}_{\tau}, u^*\in\mathcal{V}_{\tau^*}}\norm{v_0 - v_n + u+u^*}.\]
		Since the vector spaces~$\mathcal{V}_{\tau}, \mathcal{V}_{\tau^*}$, and $\mathcal{L}$ have pairwise trivial intersections, there exist unique vectors~$v\in \mathcal{V}_{\tau}, v^*\in\mathcal{V}_{\tau^*}$, and~$\bar{w}\in\mathcal{L}$, such that
		\[v_0-v_n = v+v^*+\bar{w}.\]
    With this,
    \begin{align*}
        \norm{v_0 - v_n + u+u^*}^2&=\norm{\bar{w} + (v+u)+(v^*+u^*)}^2\\
        &=\norm{\bar{w}}^2+\norm{(v+u)+(v^*+u^*)}^2,
    \end{align*}
    and thus, the minimum $\min_{u\in\mathcal{V}_{\tau}, u^*\in\mathcal{V}_{\tau^*}}\norm{v_0 - v_n + u+u^*}$ is attained for~$u=-v$ and~$u^*=-v^*$:
    \begin{align*}
        \mathrm{alt}(\tau,\tau^*) &= \min_{u\in\mathcal{V}_{\tau}, u^*\in\mathcal{V}_{\tau^*}}\norm{v_0 - v_n + u+u^*}\\
        &=\sqrt{\norm{\bar{w}}^2+\min_{u\in\mathcal{V}_{\tau}, u^*\in\mathcal{V}_{\tau^*}}\norm{(v+u)+(v^*+u^*)}^2}\\
        &=\norm{\bar{w}}\\
        &=\tfrac{1}{\norm{\bar{w}}} \scproduct{\bar{w}}{\bar{w}} =\tfrac{1}{\norm{\bar{w}}} \scproduct{v_0-v_n}{\bar{w}}.
    \end{align*}
        
		
		Any other non-zero vector $w\in \mathcal{L}$ is a scalar multiple of $\bar{w}$, which means that $\tfrac{w}{\norm{w}}=\pm \tfrac{\bar{w}}{\norm{\bar{w}}}$. Thus,
		\[\mathrm{alt}(\tau,\tau^*) =\tfrac{1}{\norm{w}} \abs{\scproduct{v_0-v_n}{w}}.\]
	\end{proof}

    \section{Relation with (generalized) cross product}
We first review some standard tools from multilinear algebra, adopting the terminology of Duistermaat and Kolk~\cite[Sec. 5.3, Remark on linear algebra]{duistermaat2004multi}. 
\begin{definition}
    For linearly independent vectors $x_1,\dots,x_{n-1}\in\R^n$, the \emph{generalized cross product}\footnote{This is also sometimes called the external product.} \newline$x_1\times\dots\times x_{n-1}$ is the unique vector such that for any vector $b\in\R^n$,
\begin{align}\label{eq:determinant}
    \left\langle b,  x_1\times\dots\times x_{n-1} \right\rangle = \det (b, x_1,\dots,x_{n-1}).
\end{align}
\end{definition}

\begin{remark}
    The generalized cross product can also be seen as the Hodge dual of the exterior product:
    \[x_1\times\dots\times x_{n-1} = \star (x_1\wedge\dots\wedge x_{n-1}).\]
\end{remark}
Equality~\ref{eq:determinant} implies that the cross product is perpendicular to all the $x_i$'s: For any $i=1,\dots,n$,
\[\left\langle x_i,  x_1\times\dots\times x_{n-1} \right\rangle = \det (x_i,\: x_1,\dots,x_i,\dots,x_{n-1})= 0.\]
The cross product thus forms a basis of the one-dimensional vector space $span \{x_1,\dots,x_{n-1}\}^\perp$.

In addition, the norm of the external product has an important geometric meaning. It is
the volume of the parallelopiped spanned by $x_1,\dots,x_{n-1}$, or, in other words, the square root of the Gram determinant of the vectors $x_1,\dots,x_{n-1}$. To be more concrete, by denoting the parallelopiped by $P(x_1,\dots,x_{n-1})$,
\[P(x_1,\dots,x_{n-1})
=
\left\{
\sum_{i=1}^{n-1} t_i x_i \;\middle|\; 0 \le t_i \le 1
\right\},\]
and the Gram matrix by $G(x_1,\dots,x_{n-1})=(\scproduct{x_i}{x_j})_{i,j=1,\dots,n-1}$, we obtain the equality
\begin{align}\label{eq:external_product_volume}
    \norm{x_1\times\dots\times x_{n-1}}=\vol_{n-1}P(x_1,\dots,x_{n-1}) = \sqrt{\det G(x_1,\dots,x_{n-1})}.
\end{align}

We connect these tools to generalized altitudes by inserting the vectors $v_1-v_0,\dots, v_k-v_0,v_{k+1}-v_n, \dots,v_{n-1}-v_n$ from Lemma~\ref{lem:altitude_expression} in the place of the vectors $x_1,\dots,x_{n-1}$.  We obtain the following corollary:

\begin{corollary}\label{cor:altitude_Gram_determinant}
Let
\[\hat{v} = (v_1-v_0)\times\dots\times (v_k-v_0)\times (v_{k+1}-v_n)\times \dots\times (v_{n-1}-v_n).\]
Then the altitude between $\tau$ and $\tau^*$ satisfies
\begin{align*}
    &\alt(\tau,\tau^*)=\frac{\abs{\scproduct{v_0-v_n}{\hat{v}}}}{\vol_{n-1}P(v_1-v_0,\dots, v_k-v_0,v_{k+1}-v_n, \dots,v_{n-1}-v_n)}.
\end{align*}
\end{corollary}
\begin{proof}
    Once we realize that the vector $\hat{v}$ is per definition perpendicular to the vectors $v_1-v_0,\dots, v_k-v_0,v_{k+1}-v_n, \dots,v_{n-1}-v_n$, the statement follows directly from the combination of Eq.~\ref{eq:altitude_angle} and~\ref{eq:external_product_volume}.
\end{proof}

\section{Approximation of altitudes by the height}
The \emph{height} of~$\sigma$ is the minimum altitude of its vertices,
	\[\height(\sigma) = \min_{i\in\{0,\dots,n\}}\alt(v_i, \sigma\backslash v_i).\]

\begin{lemma}\label{CL:altitute_bounded_by_height}
		The altitude of every face~$\tau$ of the simplex~$\sigma$ is lower bounded by the height of~$\sigma$:
        \begin{align}\label{eq:bound_with_height}
            \mathrm{alt}(\tau,\tau^*)\geq 2\:\frac{\height\left(\sigma\right)}{n+1}.
        \end{align}
	\end{lemma}

\begin{proof}
As before, let \[\mathcal{V}_{\tau} = \spann\{ v_1-v_0,\dots, v_k-v_0 \} \quad\text{and}\quad\mathcal{V}_{\tau^*} = \spann\{v_{k+1} - v_n, \dots,v_{n-1}-v_n\}.\] 
	Since~$\mathrm{alt}(\tau,\tau^*)$ is the distance between~$\mathcal{F}_\tau = v_0+\mathcal{V}_\tau$ and $\mathcal{F}_{\tau^*} = v_n+\mathcal{V}_{\tau^*}$, it is also the distance between the parallel hyperplanes
	\[\Pi_0 :=v_0 + \mathcal{V}_{\tau}+\mathcal{V}_{\tau^*} \qquad\text{ and }\qquad \Pi_n:= v_n + \mathcal{V}_{\tau}+\mathcal{V}_{\tau^*}. \]
	The altitude~$\mathrm{alt}(\tau,\tau^*)$ is thus the `thickness' of the space between the hyperplanes~$\Pi_0$ and~$\Pi_n$, or in other words, their convex hull $\conv{\{\Pi_0,\Pi_n\}}$. 
    At the same time, the hyperplane~$\Pi_0$ contains the vertices~$v_0,\dots,v_k$, and the hyperplane~$\Pi_n$ contains the vertices~$v_{k+1},\dots,v_n$. Thus, the convex hull $\conv{\{\Pi_0,\Pi_n\}}$ contains the whole simplex~$\sigma$,
	\[\conv{\{\Pi_0,\Pi_n\}}\supseteq \sigma.\]
	
	\begin{figure}[h!]
		\centering
		\includegraphics[width=0.8\textwidth]{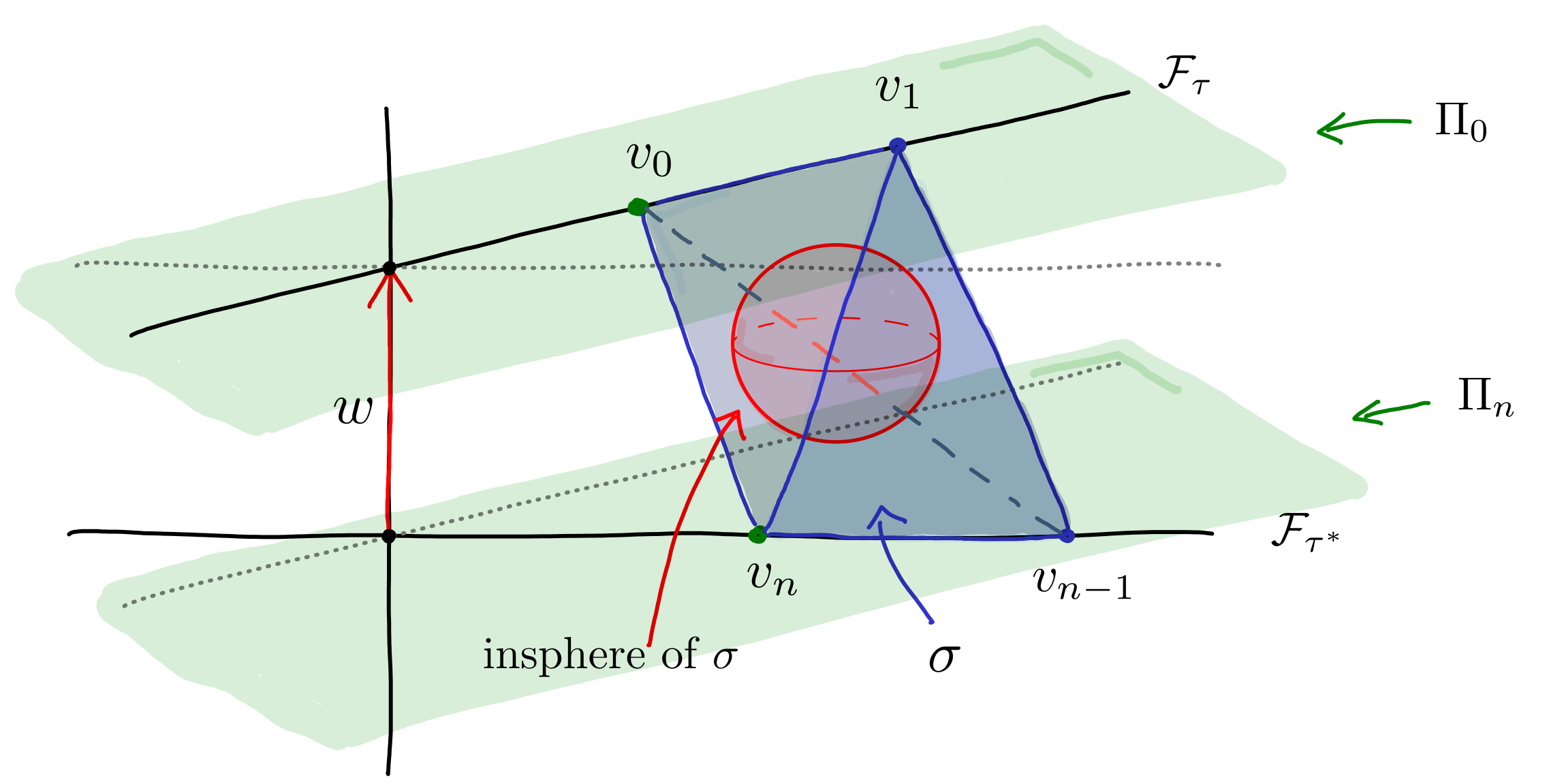}
		\caption{The convex hull of~$\Pi_0$ and $\Pi_n$ contains~$\sigma$.}
		\label{fig:division_subspaces2}
	\end{figure}
	
	In particular,~$\conv{\{\Pi_0,\Pi_n\}}$ also contains the insphere of~$\sigma$ (as illustrated in Figure~\ref{fig:division_subspaces2}). Thus, $\mathrm{alt}(\tau,\tau^*)$ is lower bounded by the diameter of the insphere of~$\sigma.$
	
	The radius of the insphere is further lower bounded by $\frac{\height\left(\sigma\right)}{n+1}.$
	Indeed, the distance of the barycentre to each of the faces of the simplex is lower bounded by $\frac{\height\left(\sigma\right)}{n+1}$. Thus, the sphere with centre at the barycentre and radius~$\frac{\height\left(\sigma\right)}{n+1}$ is contained in the interior of the simplex. Since the insphere is the largest sphere contained in the interior of a simplex, its radius must be larger than~$\frac{\height\left(\sigma\right)}{n+1}$.
\end{proof}

Given that, by definition,
\[
\alt(v_i,\sigma\backslash v_i)\geq \height(\sigma)
\]
for every vertex \(v_i\in\sigma\), it is natural to ask to what extent the constant \(\tfrac{2}{n+1}\) in bound~\eqref{eq:bound_with_height} can be improved. More concretely, what is the best constant $c$, such that 
\[\alt(\tau, \sigma\backslash \tau)\geq c\cdot\height(\sigma)\]
for all faces $\tau\subseteq \sigma$. The following example shows that this constant has to be strictly smaller than 1.

For $h>0$, consider the tetrahedron $\sigma$ with vertices
\begin{align*}
    v_0&= (0,1,0)^t & v_1&= (0,-1,0)^t & v_2&= (1 ,0,h)^t & v_3&= (-1,0,h)^t.
\end{align*}

\begin{figure}[h!]
		\centering
		\includegraphics[width=0.6\textwidth]{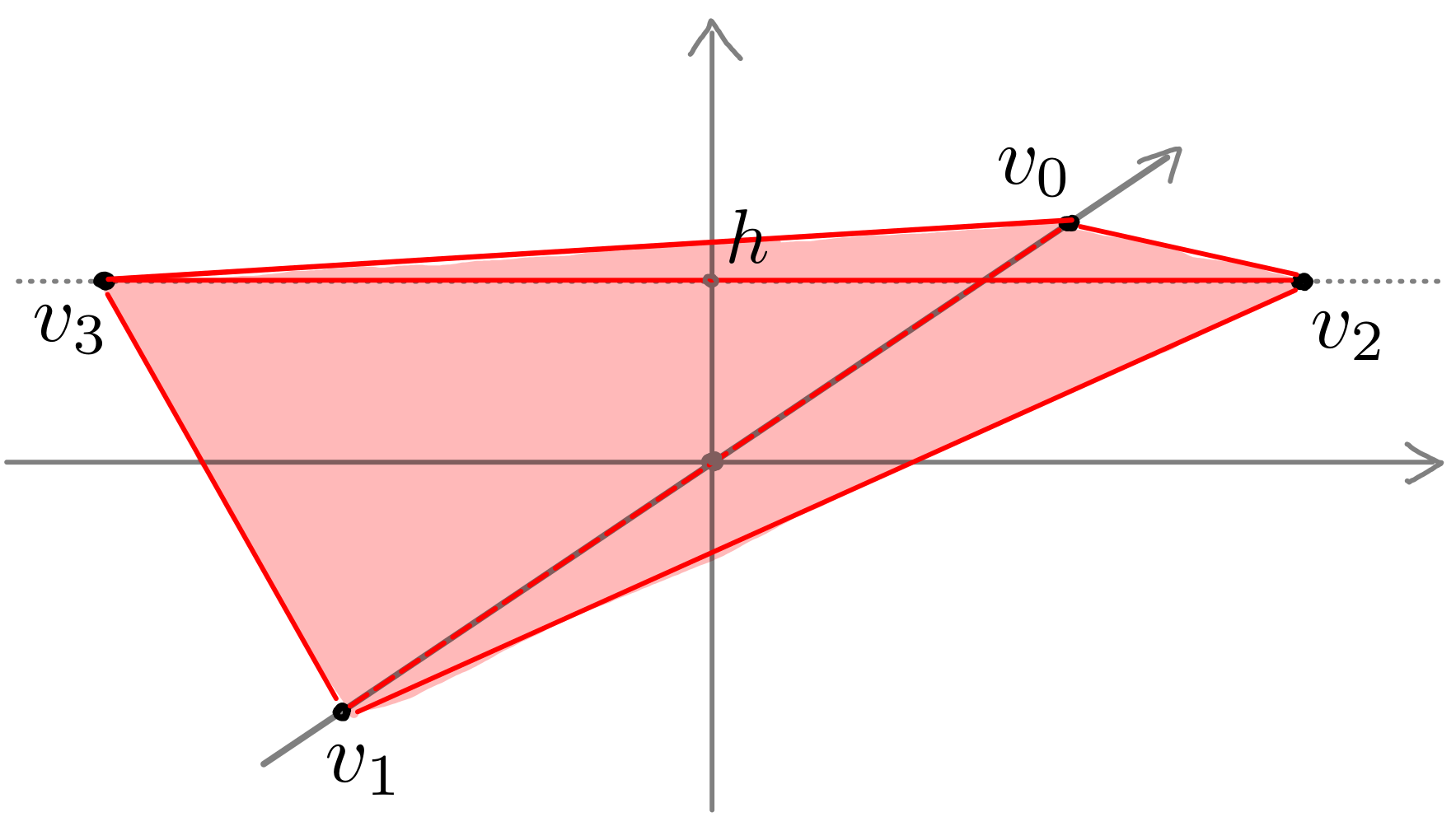}
		\caption{Illustration of the tetrahedron $\sigma$.}
		\label{fig:tetrahedron}
	\end{figure}

We depict $\sigma$ in Figure~\ref{fig:tetrahedron}. Since all four faces of $\sigma$ are congruent, the altitudes at all four vertices $v_i$ are equal,
\[
\alt(v_i, \sigma\backslash v_i)  = \height(\sigma) = \frac{2h}{\sqrt{1+h^2}}.
\]
The altitudes of the three pairs of opposite edges equal
\[
\alt(\{v_0,v_1\}, \{v_2,v_3\})  = \alt(\{v_0,v_2\}, \{v_1,v_3\}) = \sqrt{2}, \quad\alt(\{v_0,v_3\}, \{v_1,v_2\}) =h.
\]
Thus, for $0<h<\sqrt{3}$, we have that $\alt(\{v_0,v_3\}, \{v_1,v_2\})< \height(\sigma) $.

\bibliography{refs}

\end{document}